\documentclass[letterpaper, 10 pt, conference]{ieeeconf}  

\IEEEoverridecommandlockouts                              
\usepackage{cite}
\usepackage{amsmath,amssymb,amsfonts}
\usepackage{bm}
\usepackage{graphicx}
\usepackage{algorithm}
\usepackage{algpseudocode} 
\usepackage{hyperref}
\usepackage{textcomp}
\usepackage{xcolor}
\usepackage{soul}
\let\labelindent\relax
\usepackage{enumitem}

\usepackage{verbatim}
\usepackage{amsthm}
\usepackage[english]{babel}

\newtheorem{proposition}{Proposition}
\allowdisplaybreaks

\makeatletter
\g@addto@macro\normalsize{%
\setlength\abovedisplayskip{0.2pt}%
\setlength\belowdisplayskip{0.2pt}%
\setlength\abovedisplayshortskip{0.2pt}%
\setlength\belowdisplayshortskip{0.2pt}%
\setlength\belowcaptionskip{0pt}%
\setlength\abovecaptionskip{0pt}%
}
\makeatother

\title{\LARGE \bf
CAR-EnKF: A Covariance-Adaptive and Recalibrated Ensemble Kalman Filter Framework
}

\author{Shida Jiang, Shengyu Tao, Zihe Liu, and Scott Moura
\thanks{This work was not supported by any organization.}
\thanks{Shida Jiang, Zihe Liu and Scott Moura are with the Department of Civil and Environmental Engineering,
        University of California, Berkeley, CA 94720, USA. 
        {\tt\small shida\_jiang@berkeley.edu, zihe\_liu@berkeley.edu, smoura@berkeley.edu}}
\thanks{Shengyu Tao is with the Department of Electrical Engineering, Chalmers University of Technology, Gothenburg, 41296 Sweden. {\tt\small shengyu.tao@chalmers.se}}
}

\begin{document}

\maketitle
\thispagestyle{empty}
\pagestyle{empty}

\begin{abstract}
The ensemble Kalman filter (EnKF) is widely used for nonlinear and high-dimensional state estimation because it replaces explicit covariance propagation with ensemble statistics. However, conventional EnKFs can become overconfident under nonlinear measurements, while covariance inflation is only an indirect remedy. This paper proposes a covariance-adaptive and recalibrated EnKF (CAR-EnKF) that combines (i) recalibration of the covariance effect of the selected Kalman gain after the ensemble-mean update and (ii) a positive semidefinite covariance compensation driven by measurement nonlinearity. A normalized-innovation-squared feedback law adapts the compensation magnitude online. The recalibration and conditional back-out mechanisms are inherited from our prior nonlinear Kalman filter framework. This paper develops covariance-matching realizations for the stochastic EnKF and Ensemble Transform Kalman Filter (ETKF), together with EnKF-specific covariance compensation and NIS-based adaptation. The resulting corrections recover the conventional update for linear measurements. Experiments on feature-based SLAM and Lorenz--96 show lower RMSE than conventional EnKF baselines, especially at low measurement noise. The related codes are available at \href{https://github.com/Shida-Jiang/CAR-EnKF-A-Covariance-Adaptive-and-Recalibrated-Ensemble-Kalman-Filter-Framework}{https://github.com/Shida-Jiang/CAR-EnKF-A-Covariance-Adaptive-and-Recalibrated-Ensemble-Kalman-Filter-Framework}.

\end{abstract}
\section{INTRODUCTION}

Nonlinear state estimation is a fundamental problem in control, robotics, and data assimilation. Classical nonlinear Kalman-filter variants such as the extended Kalman filter (EKF) and the unscented Kalman filter (UKF) are effective for many low- and moderate-dimensional problems. However, their applicability becomes increasingly limited as the state dimension grows, since they depend on either local linearization or the propagation of a number of sigma points that scales with the state dimension \cite{julier}. In contrast, the ensemble Kalman filter (EnKF) replaces explicit covariance propagation with Monte Carlo ensemble statistics where the ensemble members can be fewer than the state dimension, making it especially attractive for high-dimensional systems and large-scale nonlinear estimation problems \cite{Evensen1994, Evensen2003, HoutekamerZhang2016}. Specifically, the EnKF has become a workhorse estimation method for numerical weather prediction, oceanographic and climate modeling, and geophysical systems \cite{HoutekamerZhang2016}.

While the first EnKF variant, the stochastic EnKF \cite{Burgers1998}, is simple and widely used, its perturbed-observation analysis introduces additional sampling noise for finite ensembles. This limitation has motivated the development of many new EnKF variants over the years. Deterministic square-root variants, including the ensemble square-root filter (EnSRF) \cite{WhitakerHamill2002}, the ensemble transform Kalman filter (ETKF) \cite{Bishop2001}, and the local ensemble transform Kalman filter (LETKF) \cite{Hunt2007} reduce this effect while retaining the computational appeal that makes EnKF methods standard tools for nonlinear and high-dimensional estimation \cite{HoutekamerZhang2016, Sakov2008}.

Despite their success, conventional EnKFs often exhibit covariance underestimation in problems with nonlinear measurements and limited ensemble sizes. In practice, covariance inflation is often introduced to improve robustness, but it is usually applied as a generic empirical correction that does not account for the local nonlinearity of the measurement operator \cite{Anderson2007, Luo2013, ElGharamti2018}. This is particularly unsatisfactory because, when the measurement operator is highly nonlinear, the conventional EnKF update can become systematically overconfident. In contrast, when the measurements are locally close to linear, this issue becomes much less pronounced.

In this paper, we propose a covariance-adaptive and recalibrated ensemble Kalman filter (CAR-EnKF) framework for nonlinear state estimation. It combines (i) recalibration of the covariance effect of a chosen gain and conditional back-out, inherited from our prior nonlinear-Kalman-filter framework \cite{frame}, with (ii) a positive semidefinite covariance compensation driven by measurement nonlinearity. The framework is applicable to multiple EnKF variants through variant-specific anomaly updates and reduces to the conventional, non-inflated update for linear measurements. We specialize it to the stochastic EnKF and the ETKF and compare each realization with its conventional fixed-inflation counterpart. The main contributions are as follows.
\begin{enumerate}
    \item We translate the prior recalibration and back-out principle into ensemble form and derive stochastic-EnKF and ETKF anomaly updates that match the recalibrated covariance target in expectation and exactly, respectively.
    \item Inspired by the covariance-compensation principles for other nonlinear KFs \cite{compensate}, we develop an EnKF-specific, nonlinearity-dependent compensation and an optional normalized-innovation-squared adaptation law. The compensation is positive semidefinite and vanishes for linear measurements.
    \item We evaluate the two realizations on feature-based SLAM and Lorenz--96, where they reduce state RMSE by up to several orders of magnitude relative to the tested conventional fixed-inflation baselines across a range of measurement-noise levels.
\end{enumerate}

The paper is organized as follows. Section \ref{sec:method} details the conventional and proposed EnKF frameworks. Section \ref{sec:exp} presents numerical results. Section \ref{sec:conclusion} concludes the paper.

\section{METHOD}
\label{sec:method}

We consider the discrete-time nonlinear system
\begin{align}
    \bm{x}_k &= f_{k-1}(\bm{x}_{k-1}, \bm{u}_{k-1}) + \bm{w}_{k-1}, \\
    \bm{z}_k &= h_k(\bm{x}_k) + \bm{v}_k,
\end{align}
where $\bm{x}_k \in \mathbb{R}^n$ is the state, $\bm{z}_k \in \mathbb{R}^m$ is the measurement, $\bm{w}_{k-1}$ is the process noise, and $\bm{v}_k$ is the measurement noise. Let the ensemble size be $N$.

\subsection{Conventional EnKF framework}
\label{subsec:conventional_enkf}

At time step $k-1$, suppose the analysis (posterior) ensemble is
\begin{equation}
    X_{k-1}^a
    =
    \begin{bmatrix}
        \bm{x}_{k-1,1}^a & \cdots & \bm{x}_{k-1,N}^a
    \end{bmatrix}.
\end{equation}
Each ensemble member is propagated through the dynamics to form the
pre-inflation forecast ensemble
\begin{equation}
    \bm{x}_{k,i}^{f_0}
    =
    f_{k-1}(\bm{x}_{k-1,i}^a,\bm{u}_{k-1}) + \bm{w}_{k-1,i},
    \quad i=1,\dots,N,
\end{equation}
where \(\bm{w}_{k-1,i}\) denotes the process-noise realization associated with the 
$i$th ensemble member, independently drawn from $\mathcal{N}(\bm{0},Q_{k-1})$. Denote the forecast mean and anomaly matrix by
\begin{align}
    \bar{\bm{x}}_k^f = \frac{1}{N}\sum_{i=1}^N \bm{x}_{k,i}^{f_0}, \quad
    A_k^{f_0} = X_k^{f_0} - \bar{\bm{x}}_k^{f} \bm{1}^\top,
\end{align}
where
\begin{equation}
    X_k^{f_0}
    =
    \begin{bmatrix}
        \bm{x}_{k,1}^{f_0} & \cdots & \bm{x}_{k,N}^{f_0}
    \end{bmatrix}.
\end{equation}
In conventional EnKF implementations, multiplicative covariance inflation is often applied to the forecast anomalies:
\begin{equation}\label{inflation_eq}
    A_k^f = \sqrt{\rho}\, A_k^{f_0},
    \quad \rho \ge 1.
\end{equation}
After covariance inflation, the forecast ensemble becomes
\begin{equation}
    X_k^{f}
    =
    A_k^{f} + \bar{\bm{x}}_k^{f} \bm{1}^\top,
\end{equation}
and the sample forecast covariance is
\begin{equation}\label{predict}
    P_k^f = \frac{1}{N-1} A_k^f (A_k^f)^\top.
\end{equation}

Each forecast member is then mapped to the measurement space:
\begin{equation}
    \bm{z}_{k,i}^f = h_k(\bm{x}_{k,i}^f), \quad i=1,\dots,N.
\end{equation}
Let $Z_k^f$ denote the raw forecast-measurement ensemble and $\mathcal{Z}_k^f$ its centered anomaly matrix:
\begin{align}
    Z_k^f &=
    \begin{bmatrix}
        \bm{z}_{k,1}^f & \cdots & \bm{z}_{k,N}^f
    \end{bmatrix}, \\
    \bar{\bm{z}}_k^f &= \frac{1}{N}\sum_{i=1}^N \bm{z}_{k,i}^f, \quad
    \mathcal{Z}_k^f = Z_k^f - \bar{\bm{z}}_k^f \bm{1}^\top.
\end{align}
The forecast state--measurement cross-covariance and forecast measurement covariance are estimated by
\begin{align}
    P_{xz,k}^f = \frac{1}{N-1} A_k^f (\mathcal{Z}_k^f)^\top, \quad
    P_{zz,k}^f = \frac{1}{N-1} \mathcal{Z}_k^f (\mathcal{Z}_k^f)^\top.
\end{align}
The innovation covariance and Kalman gain are
\begin{equation}
    S_k^f = P_{zz,k}^f + R_k,
\end{equation}
\begin{equation}\label{gain}
    K_k = P_{xz,k}^f (S_k^f)^{-1}.
\end{equation}

The final step of an EnKF is to update the ensemble members according to the Kalman gain $K_k$ so that the updated ensemble mean ($\bar{\bm{x}}_k^a$) and covariance ($P_k^a$) (approximately) satisfy
\begin{equation}\label{hope1}
    \bar{\bm{x}}_k^a
    =
    \bar{\bm{x}}_k^f
    +
    K_k \bigl(\bm{z}_k - \bar{\bm{z}}_k^f\bigr).
\end{equation}
\begin{equation}\label{hope2}
    P_k^a
    =
    P_k^f
    -
    K_k S_k^fK_k^\top.
\end{equation}
To achieve this goal, different EnKFs update the ensemble members in different ways. Below, we discuss two representative variants: stochastic EnKF and ETKF.
\subsubsection{Conventional stochastic EnKF}\label{subsubsec:stoch_enkf}
The stochastic EnKF follows the perturbed-observation formulation in \cite{Burgers1998}. After computing $K_k$, one draws independent measurement perturbations
\begin{equation}
    \bm{\eta}_{k,i} \sim \mathcal{N}(\bm{0}, R_k), \quad i=1,\dots,N,
\end{equation}
and updates each member by
\begin{equation}
    \bm{\widetilde{x}}_{k,i}^a
    =
    \bm{x}_{k,i}^f
    +
    K_k \bigl(\bm{z}_k + \bm{\eta}_{k,i} - \bm{z}_{k,i}^f \bigr).
    \label{eq:stoch_update}
\end{equation}
Hence, the stochastic EnKF uses a memberwise random analysis update.
Define the analysis sample mean and sample covariance by
\begin{equation}
    \bar{\widetilde{\bm{x}}}_k^a
    =
    \frac{1}{N}\sum_{i=1}^N \widetilde{\bm{x}}_{k,i}^a, \quad
    \widetilde{P}_k^a
    =
    \frac{1}{N-1}\widetilde{A}_k^a(\widetilde{A}_k^a)^\top,
\end{equation}
where $\widetilde{A}_k^a$ is the anomaly matrix of the updated ensemble.
Then, one can show that
\begin{equation}
    \mathbb{E}[\bar{\widetilde{\bm{x}}}_k^a]
    =
    \bar{\bm{x}}_k^f
    +
    K_k \bigl(\bm{z}_k - \bar{\bm{z}}_k^f\bigr)
    =
    \bar{\bm{x}}_k^a,
\end{equation}
and
\begin{equation}
    \mathbb{E}[\widetilde{P}_k^a]
    =
    P_k^f
    -
    K_k S_k^f K_k^\top
    =
    P_k^a,
\end{equation}
where the expectation is taken with respect to
$\bm{\eta}_{k,1},\ldots,\bm{\eta}_{k,N}$.
Thus, for stochastic EnKF, \eqref{hope1} and \eqref{hope2} are approximately satisfied in an expectation sense.

\subsubsection{Conventional ETKF}
\label{subsubsec:etkf}
The ETKF follows a deterministic transform formulation \cite{Bishop2001}. The ensemble mean is directly updated by \eqref{hope1}, while the anomalies are updated through
\begin{equation}\label{Told}
    A_k^a = A_k^f T_k,
\end{equation}
where $T_k \in \mathbb{R}^{N\times N}$ is chosen so that the sample covariance matches the Kalman posterior covariance shown in \eqref{hope2} \cite{wang2003comparison}. One convenient symmetric choice is
\begin{equation}\label{conventional}
    T_k
    =
    \Big(I
    -
    (\mathcal{Z}_k^f)^\top
    \Bigl(
    \mathcal{Z}_k^f(\mathcal{Z}_k^f)^\top + (N-1)R_k
    \Bigr)^{-1}
    \mathcal{Z}_k^f \Big)^{1/2}.
\end{equation}
The final analysis ensemble is
\begin{equation}
    X_k^a = \bar{\bm{x}}_k^a \bm{1}^\top + A_k^a.
\end{equation}
\subsection{Improvement 1: Recalibration and Conditional Back-out}
\label{subsec:improvement1}
The recalibration target and trace-based back-out rule are inherited from our prior nonlinear-Kalman-filter framework \cite{frame}, whose Theorem 1 establishes the covariance-underestimation result under its stated approximation assumptions. For a fixed gain, posterior covariance has the generalized Joseph form $P^f+KSK^\top-KP_{zx}-P_{xz}K^\top$. The conventional reduction \eqref{hope2} can be optimistic when approximate moments selected before a nonlinear mean shift no longer describe the updated region. Recalibration therefore recomputes the measurement-space moments around the updated mean. The new contribution here is to construct stochastic-EnKF and ETKF ensembles whose sample covariances match this target in expectation and exactly, respectively.

If the recalibrated posterior covariance has a larger trace than the forecast covariance, the conditional back-out step reverts to the forecast. Consequently, the trace of the covariance used by the filter does not increase during the measurement update. Algorithm~\ref{new_algorithm} summarizes the inherited framework.

\begin{algorithm}[htbp]
\footnotesize
\caption{The nonlinear Kalman filter framework with recalibration and conditional back-out proposed in \cite{frame}}\label{new_algorithm}
\begin{algorithmic}[1]
    \Statex \textbf{Input:} Process-noise covariance $Q_{k-1}$, measurement-noise covariance $R_k$, state-transition function $f_{k-1}$, measurement function $h_k$, input $\bm{u}_{k-1}$, measurement $\bm{z}_k$

    \Statex \textbf{Initialization:}
    \State Initialize the analysis mean $\bar{\bm{x}}_0^a = \mathbb{E}[\bm{x}_0]$
    \State Initialize the analysis covariance $P_0^a = \operatorname{Cov}(\bm{x}_0)$

    \For{$k=1,2,\dots$}
        \Statex \hspace{1.2em} \textbf{Predict:}
        \State Compute the forecast mean $\bar{\bm{x}}_k^f$ and forecast covariance $P_k^f$
        \Statex \hspace{1.2em} from $\bar{\bm{x}}_{k-1}^a$, $P_{k-1}^a$, $f_{k-1}$, $\bm{u}_{k-1}$, and $Q_{k-1}$
        \Statex \hspace{1.2em} \textbf{Update:}
        \State Compute the predicted measurement mean $\bar{\bm{z}}_k^f$, the forecast 
        \Statex \hspace{1.2em} state--measurement cross-covariance $P_{xz,k}^f$, and the 
        \Statex \hspace{1.2em} forecast measurement covariance $P_{zz,k}^f$ from $\bar{\bm{x}}_k^f$, $P_k^f$, and $h_k$
        \State Set $S_k^f = P_{zz,k}^f + R_k$
        \State Compute the Kalman gain $K_k = P_{xz,k}^f (S_k^f)^{-1}$
        \State Update the analysis mean $\bar{\bm{x}}_k^a=\bar{\bm{x}}_k^f + K_k(\bm{z}_k - \bar{\bm{z}}_k^f)$

        \Statex \hspace{1.2em} \textbf{Recalibrate:}
        \State Compute the recalibrated cross-covariance $P_{xz,k}^{rc}$ and recalibrated 
        \Statex \hspace{1.2em} measurement covariance $P_{zz,k}^{rc}$ from $\bar{\bm{x}}_k^a$, $P_k^f$, and $h_k$
        \State Set $S_k^{rc} = P_{zz,k}^{rc} + R_k$
        \State Form the recalibrated posterior covariance
        \Statex \hspace{1.2em} $P_k^{a,rc}
            =
            P_k^f
            +
            K_k S_k^{rc} K_k^\top
            -
            P_{xz,k}^{rc} K_k^\top
            -
            K_k (P_{xz,k}^{rc})^\top$
        \Statex \hspace{1.2em} \textbf{Back out:}
        \If{$\operatorname{tr}(P_k^{a,rc}) > \operatorname{tr}(P_k^f)$}
            \State Set $\bar{\bm{x}}_k^a = \bar{\bm{x}}_k^f$
            \State Set $P_k^a = P_k^f$
        \Else
            \State Set $P_k^a = P_k^{a,rc}$
        \EndIf
    \EndFor
\end{algorithmic}
\normalsize
\end{algorithm}

Following Algorithm~\ref{new_algorithm}, an EnKF requires an ensemble-based recalibrated covariance evaluation and a variant-specific anomaly update, which are developed below for the stochastic EnKF and ETKF.

After the Kalman gain is computed using \eqref{gain} and the ensemble mean is updated using \eqref{hope1}, we can form the recentered (recalibrated) ensemble
\begin{equation}
    X_k^{rc}
    =
    \begin{bmatrix}
        \bm{x}_{k,1}^{rc} & \cdots & \bm{x}_{k,N}^{rc}
    \end{bmatrix}
    =
    \bar{\bm{x}}_k^a \bm{1}^\top + A_k^f.
\end{equation}
Let the corresponding recalibrated predicted measurements be
\begin{equation}
    \bm{z}_{k,i}^{rc} = h_k(\bm{x}_{k,i}^{rc}),
    \quad i=1,\dots,N,
\end{equation}
and define
\small\begin{align}
    \bar{\bm{z}}_k^{rc} = \frac{1}{N}\sum_{i=1}^N \bm{z}_{k,i}^{rc}, \quad
    \mathcal{Z}_k^{rc}
    =
    \begin{bmatrix}
        \bm{z}_{k,1}^{rc} & \cdots & \bm{z}_{k,N}^{rc}
    \end{bmatrix}
    - \bar{\bm{z}}_k^{rc}\bm{1}^\top.
\end{align}\normalsize
The cross-covariance can then be recalibrated as
\begin{equation}\label{cross}
    P_{xz,k}^{rc} = \frac{1}{N-1} A_k^f (\mathcal{Z}_k^{rc})^\top.
\end{equation}
Similarly, the innovation covariance can be recalibrated as
\begin{equation}\label{ETKF_target_new}
    S_k^{rc}
    =
    \frac{1}{N-1}\mathcal{Z}_k^{rc}(\mathcal{Z}_k^{rc})^\top
    +
    R_k.
\end{equation}
This completes steps 9 and 10 in Algorithm \ref{new_algorithm}. Then, according to step 11, the resulting posterior covariance target is
\begin{equation}\label{hope3}
    P_k^{a,rc}
    =
    P_k^f
    +
    K_k S_k^{rc} K_k^\top
    -
    K_k (P_{xz,k}^{rc})^\top
    -
    P_{xz,k}^{rc} K_k^\top.
\end{equation}
Note that \eqref{hope3} reduces exactly to \eqref{hope2} when the measurement model is linear. Therefore, the proposed improvement is specifically intended to address nonlinear measurements. To enforce \eqref{hope3}, the anomaly update rule must be modified for different EnKF variants. Again, we discuss two representative variants: stochastic EnKF and ETKF.

\subsubsection{Improvement 1 for stochastic EnKF}

The key idea of the stochastic EnKF is to perform member-wise updates with independently perturbed observations. Following this idea, we independently draw perturbations
\begin{equation}
    \bm{\eta}_{k,i}\sim\mathcal{N}(\bm{0},R_k), \quad i=1,\dots,N,
\end{equation}
and define the temporary updated ensemble by
\begin{equation}
    \bm{x}_{k,i}^{\star}
    =
    \bm{x}_{k,i}^{rc}
    +
    K_k
    \bigl(
        \bm{z}_k + \bm{\eta}_{k,i} - \bm{z}_{k,i}^{rc}
    \bigr).
    \label{eq:i1_stoch_temp}
\end{equation}
Let
\begin{equation}\label{Xk}
    X_k^\star
    =
    \begin{bmatrix}
        \bm{x}_{k,1}^\star & \cdots & \bm{x}_{k,N}^\star
    \end{bmatrix},
\end{equation}
with sample mean
\begin{equation}
    \bar{\bm{x}}_k^\star
    =
    \frac{1}{N}\sum_{i=1}^N \bm{x}_{k,i}^\star,
\end{equation}
and anomaly matrix
\begin{equation}
    A_k^\star
    =
    X_k^\star - \bar{\bm{x}}_k^\star \bm{1}^\top.
\end{equation}
Because the finite-ensemble sample mean of \eqref{eq:i1_stoch_temp} is not exactly $\bar{\bm{x}}_k^a$, we translate the temporary ensemble once more and define the candidate analysis ensemble
\begin{equation}\label{translate}
    \widetilde{X}_k^a
    =
    \bar{\bm{x}}_k^a \bm{1}^\top + A_k^\star.
\end{equation}
Hence, the realized sample covariance of the candidate analysis ensemble is
\begin{equation}\label{eqn_EnKF}
    \widetilde{P}_k^{a,rc}
    =
    \frac{1}{N-1} A_k^\star (A_k^\star)^\top.
\end{equation}
Following Algorithm~\ref{new_algorithm}, one may accept the update only if
\begin{equation}
    \operatorname{tr}(P_k^{a,rc})
    \le
    \operatorname{tr}(P_k^f),
\end{equation}
in which case $X_k^a=\widetilde{X}_k^a$; otherwise, the filter backs out and sets
\begin{equation}
    X_k^a = X_k^f.
\end{equation}
This update rule guarantees that \eqref{hope3} is satisfied in the expectation sense, as stated in the following proposition.

\begin{proposition}\label{prop:i1_stoch}
Assume the recalibrated stochastic-EnKF update is accepted. Then the candidate
analysis ensemble $\widetilde{X}_k^a$ satisfies
\begin{equation}\label{proof1}
    \frac{1}{N}\widetilde{X}_k^a \bm{1}
    =
    \bar{\bm{x}}_k^a
\end{equation}
exactly, and its sample covariance satisfies
\begin{equation}\label{eq:cov_expectation}
    \mathbb{E}\!\left[\widetilde{P}_k^{a,rc}\right]
    =
    P_k^{a,rc},
\end{equation}
where $\widetilde{P}_k^{a,rc}$ is defined in \eqref{eqn_EnKF}, and the expectation is taken with respect to $\bm{\eta}_{k,1},\ldots,\bm{\eta}_{k,N}$.
\end{proposition}

\begin{proof}
Equation \eqref{proof1} follows immediately from \eqref{translate}, since
\begin{equation}
    \frac{1}{N}\widetilde{X}_k^a\bm{1}
    =
    \frac{1}{N}
    \left(
        \bar{\bm{x}}_k^a \bm{1}^\top + A_k^\star
    \right)\bm{1}
    =
    \bar{\bm{x}}_k^a
    +
    \frac{1}{N}A_k^\star \bm{1}
    =
    \bar{\bm{x}}_k^a,
\end{equation}
where the last equality holds because $A_k^\star$ is an anomaly matrix, and therefore $A_k^\star \bm{1}=\bm{0}$.

Next, define the centered perturbation matrix
\begin{equation}
    \mathcal{E}_k
    =
    \begin{bmatrix}
        \bm{\eta}_{k,1} & \cdots & \bm{\eta}_{k,N}
    \end{bmatrix}
    -
    \bar{\bm{\eta}}_k \bm{1}^\top,
\end{equation}
where
\begin{equation}
    \bar{\bm{\eta}}_k
    =
    \frac{1}{N}\sum_{i=1}^N \bm{\eta}_{k,i}.
\end{equation}
Since $\bm{x}_{k,i}^{rc} - \bar{\bm{x}}_k^a$ is the $i$th column of $A_k^f$, and $\bm{z}_{k,i}^{rc} - \bar{\bm{z}}_k^{rc}$ is the $i$th column of $\mathcal{Z}_k^{rc}$, the anomaly matrix of the temporary ensemble can be written as
\begin{equation}\label{eq:Astar_expand}
    A_k^\star
    =
    A_k^f
    -
    K_k \mathcal{Z}_k^{rc}
    +
    K_k \mathcal{E}_k.
\end{equation}
Substituting \eqref{eq:Astar_expand} into \eqref{eqn_EnKF} gives
\footnotesize\begin{align}
    \widetilde{P}_k^{a,rc}
    &=
    \frac{1}{N-1}
    \left(
        A_k^f - K_k \mathcal{Z}_k^{rc} + K_k \mathcal{E}_k
    \right)
    \left(
        A_k^f - K_k \mathcal{Z}_k^{rc} + K_k \mathcal{E}_k
    \right)^\top.
\end{align}\normalsize
Taking expectation with respect to $\bm{\eta}_{k,1},\ldots,\bm{\eta}_{k,N}$, and noting that $A_k^f$ and $\mathcal{Z}_k^{rc}$ are deterministic conditioned on the forecast ensemble, while
\begin{equation}
    \mathbb{E}[\mathcal{E}_k] = 0,
    \quad
    \frac{1}{N-1}\mathbb{E}[\mathcal{E}_k \mathcal{E}_k^\top] = R_k,
\end{equation}
we obtain
\footnotesize\begin{align}
    &\mathbb{E}\!\left[\widetilde{P}_k^{a,rc}\right]
    =
    \frac{1}{N-1}A_k^f (A_k^f)^\top
    -
    \frac{1}{N-1}A_k^f (\mathcal{Z}_k^{rc})^\top K_k^\top
    \nonumber\\
    &\quad
    -
    \frac{1}{N-1}K_k \mathcal{Z}_k^{rc} (A_k^f)^\top
    +
    K_k
    \left(
        \frac{1}{N-1}\mathcal{Z}_k^{rc}(\mathcal{Z}_k^{rc})^\top + R_k
    \right)
    K_k^\top.
\end{align}\normalsize
Using \eqref{predict}, \eqref{cross}, and \eqref{ETKF_target_new}, it follows that
\begin{equation}
    \mathbb{E}\!\left[\widetilde{P}_k^{a,rc}\right]
    =
    P_k^f
    +
    K_k S_k^{rc} K_k^\top
    -
    K_k (P_{xz,k}^{rc})^\top
    -
    P_{xz,k}^{rc} K_k^\top.
\end{equation}
By \eqref{hope3}, this is exactly $P_k^{a,rc}$. Therefore, \eqref{eq:cov_expectation} holds.
\end{proof}

\subsubsection{Improvement 1 for ETKF}
The key idea of ETKF is to find a transformation matrix $T_k^{rc}$ such that
the updated anomalies
\begin{equation}\label{updated_A}
    \widetilde{A}_k^a = A_k^f T_k^{rc},
\end{equation}
have a sample covariance equal to $P_k^{a,rc}$ defined in \eqref{hope3}. The candidate analysis ensemble is then
\begin{equation}
    \widetilde{X}_k^a = \bar{\bm{x}}_k^a \bm{1}^\top + \widetilde{A}_k^a.
\end{equation}
One choice of $T_k^{rc}$ is the symmetric positive-semidefinite square root
\footnotesize\begin{equation}\label{eq:mkrc}
    T_k^{rc}
    =
    \Big(
    I
    -
    (\mathcal{Z}_k^f)^\top B_k^f \mathcal{Z}_k^{rc}
    -
    (\mathcal{Z}_k^{rc})^\top B_k^f \mathcal{Z}_k^f
    +
    (\mathcal{Z}_k^f)^\top B_k^f \Gamma_k^{rc} B_k^f \mathcal{Z}_k^f
    \Big)^{1/2},
\end{equation}\normalsize
where
\begin{equation}
    B_k^f
    =
    \Big(
    \mathcal{Z}_k^f(\mathcal{Z}_k^f)^\top
    +
    (N-1)R_k
    \Big)^{-1},
\end{equation}
and
\begin{equation}
    \Gamma_k^{rc}
    =
    \mathcal{Z}_k^{rc}(\mathcal{Z}_k^{rc})^\top
    +
    (N-1)R_k.
\end{equation}
Indeed, the matrix inside the square root is positive semidefinite, because
\begin{align}
&\;
I
-
(\mathcal{Z}_k^f)^\top B_k^f \mathcal{Z}_k^{rc}
-
(\mathcal{Z}_k^{rc})^\top B_k^f \mathcal{Z}_k^f
+
(\mathcal{Z}_k^f)^\top B_k^f \Gamma_k^{rc} B_k^f \mathcal{Z}_k^f
\notag\\
&\quad=
\bigl(I-(\mathcal{Z}_k^{rc})^\top B_k^f \mathcal{Z}_k^f\bigr)^\top
\bigl(I-(\mathcal{Z}_k^{rc})^\top B_k^f \mathcal{Z}_k^f\bigr)
\notag\\
&\quad+
(N-1)(\mathcal{Z}_k^f)^\top B_k^f R_k B_k^f \mathcal{Z}_k^f
\succeq 0.
\end{align}
Following Algorithm~\ref{new_algorithm}, one may accept the update and set $X_k^a=\widetilde{X}_k^a$ only if
\begin{equation}
    \operatorname{tr}(P_k^{a,rc})
    \le
    \operatorname{tr}(P_k^f),
\end{equation}
otherwise the filter reverts to the forecast ensemble and sets
\begin{equation}
    X_k^a = X_k^f.
\end{equation}
This update rule guarantees that \eqref{hope3} is exactly satisfied, as stated in the following proposition.

\begin{proposition}\label{p2}
Assume the recalibrated ETKF update is accepted, and let the updated anomalies be defined by \eqref{updated_A}, where $T_k^{rc}$ is given by \eqref{eq:mkrc}. Then, the realized sample covariance of the candidate analysis ensemble satisfies
\begin{equation}\label{proof2}
    \frac{1}{N-1}\widetilde{A}_k^a(\widetilde{A}_k^a)^\top = P_k^{a,rc}.
\end{equation}
\end{proposition}

\begin{proof}
By \eqref{updated_A},
\begin{equation}
    \frac{1}{N-1}\widetilde{A}_k^a(\widetilde{A}_k^a)^\top
    =
    \frac{1}{N-1}A_k^f T_k^{rc}(T_k^{rc})^\top (A_k^f)^\top.
\end{equation}
Since $T_k^{rc}$ is the symmetric square root in \eqref{eq:mkrc},
\begin{align}
\frac{1}{N-1}\widetilde{A}_k^a(\widetilde{A}_k^a)^\top
&=
\frac{1}{N-1}A_k^f(A_k^f)^\top
\nonumber\\
&\quad
-
\frac{1}{N-1}A_k^f(\mathcal{Z}_k^f)^\top B_k^f \mathcal{Z}_k^{rc}(A_k^f)^\top
\nonumber\\
&\quad
-
\frac{1}{N-1}A_k^f(\mathcal{Z}_k^{rc})^\top B_k^f \mathcal{Z}_k^f (A_k^f)^\top
\nonumber\\
&\quad
+
\frac{1}{N-1}A_k^f(\mathcal{Z}_k^f)^\top
B_k^f \Gamma_k^{rc} B_k^f \mathcal{Z}_k^f (A_k^f)^\top.
\label{eq:p1_expand}
\end{align}
Now recall that
\begin{equation}
    P_k^f = \frac{1}{N-1}A_k^f(A_k^f)^\top,
    \quad
    P_{xz,k}^{rc} = \frac{1}{N-1}A_k^f(\mathcal{Z}_k^{rc})^\top,
\end{equation}
and
\small\begin{equation}
    K_k
    =
    A_k^f(\mathcal{Z}_k^f)^\top
    \Bigl(
    \mathcal{Z}_k^f(\mathcal{Z}_k^f)^\top + (N-1)R_k
    \Bigr)^{-1}
    =
    A_k^f(\mathcal{Z}_k^f)^\top B_k^f.
\end{equation}\normalsize
Moreover, from the definition of $\Gamma_k^{rc}$,
\begin{equation}
    \Gamma_k^{rc}
    =
    \mathcal{Z}_k^{rc}(\mathcal{Z}_k^{rc})^\top + (N-1)R_k
    =
    (N-1)S_k^{rc}.
\end{equation}
Therefore, the second term on the right-hand side of \eqref{eq:p1_expand} becomes
\begin{equation}
    \frac{1}{N-1}A_k^f(\mathcal{Z}_k^f)^\top B_k^f \mathcal{Z}_k^{rc}(A_k^f)^\top
    =
    K_k (P_{xz,k}^{rc})^\top,
\end{equation}
the third term becomes
\begin{equation}
    \frac{1}{N-1}A_k^f(\mathcal{Z}_k^{rc})^\top B_k^f \mathcal{Z}_k^f (A_k^f)^\top
    =
    P_{xz,k}^{rc} K_k^\top,
\end{equation}
and the fourth term becomes
\begin{equation}
    \frac{1}{N-1}A_k^f(\mathcal{Z}_k^f)^\top
    B_k^f \Gamma_k^{rc} B_k^f \mathcal{Z}_k^f (A_k^f)^\top
    =
    K_k S_k^{rc} K_k^\top.
\end{equation}
Substituting these identities into \eqref{eq:p1_expand} yields
\small\begin{equation}
    \frac{1}{N-1}\widetilde{A}_k^a(\widetilde{A}_k^a)^\top
    =
    P_k^f
    +
    K_k S_k^{rc} K_k^\top
    -
    K_k (P_{xz,k}^{rc})^\top
    -
    P_{xz,k}^{rc} K_k^\top.
\end{equation}\normalsize
By \eqref{hope3}, this is exactly $P_k^{a,rc}$. Therefore, \eqref{proof2} holds.
\end{proof}

\subsection{Improvement 2: Adaptive Covariance Compensation}
\label{subsec:improvement2}

As presented in \eqref{inflation_eq}, the conventional EnKF typically applies multiplicative covariance inflation to the forecast anomalies. Although this often improves performance empirically, its mathematical justification is much weaker for linear measurements. Indeed, if
\begin{equation}
    h_k(\bm{x}_k) = H_k \bm{x}_k,
\end{equation}
then
\begin{equation}
    \bar{\bm{z}}_k^f = H_k \bar{\bm{x}}_k^f,
    \quad
    \mathcal{Z}_k^f = H_k A_k^f,
\end{equation}
and hence
\small\begin{equation}
    P_{xz,k}^f = P_k^f H_k^\top,
    \quad
    P_{zz,k}^f = H_k P_k^f H_k^\top,
    \quad
    S_k^f = H_k P_k^f H_k^\top + R_k.
\end{equation}\normalsize
Therefore, if the forecast mean and covariance are exact, then the corresponding first and second moments in measurement space are also exact. Thus, linear measurements introduce no additional approximation at the level of first and second moments. Although nonlinear state dynamics may still render the forecast distribution non-Gaussian, there is no general result showing that this mismatch systematically causes the forecast covariance to underestimate the true state covariance.

This point is further supported by the prediction-only setting. If the initial ensemble members and process-noise samples are i.i.d. draws from their true distributions, then the pre-inflation forecast covariance
\begin{equation}
    P_k^{f_0}
    =
    \frac{1}{N-1}A_k^{f_0}(A_k^{f_0})^\top
\end{equation}
satisfies
\begin{equation}
    \mathbb{E}[P_k^{f_0}] = \operatorname{Cov}(\bm{x}_k).
\end{equation}
After multiplicative inflation in \eqref{inflation_eq}, we obtain
\begin{equation}
    P_k^f = \rho P_k^{f_0},
    \quad 
    \mathbb{E}[P_k^f] = \rho \, \operatorname{Cov}(\bm{x}_k), \quad\rho \ge 1,
\end{equation}
which is biased whenever \(\rho>1\). Thus, for linear measurements, covariance inflation is better viewed as a pragmatic robustness device than as a mathematically necessary correction. By contrast, when \(h_k\) is nonlinear, the matrices \(P_{xz,k}^f\) and \(P_{zz,k}^f\) depend on higher-order properties of the forecast distribution and cannot, in general, be recovered exactly from the first two moments alone. In that case, Theorem 1 in \cite{frame} shows that even unbiased approximations of the innovation covariance and cross-covariance can still cause the conventional covariance update to underestimate the actual posterior covariance on average. This observation motivates the second improvement proposed in this paper -- covariance compensation -- which enlarges the innovation covariance for nonlinear measurements while vanishing in the linear case.

Improvement~2 removes the covariance inflation step (by fixing $\rho=1$) and augments the innovation covariance by a positive semidefinite compensation term that depends on the mismatch between $h_k(\bar{\bm{x}}_k^f)$ and $\bar{\bm{z}}_k^f$. Specifically, the compensated forecast innovation covariance is
\begin{equation}
    \widetilde{S}_k^f
    =
    \frac{1}{N-1}\mathcal{Z}_k^f(\mathcal{Z}_k^f)^\top
    +
    \beta_{k-1}\bm{d}_k^f(\bm{d}_k^f)^\top
    +
    R_k,
\end{equation}
where $\beta_{k-1}$ is a scalar and $\bm{d}_k^f$ is the forecast nonlinearity mismatch defined by
\begin{equation}\label{dkf}
    \bm{d}_k^f
    =
    h_k(\bar{\bm{x}}_k^f) - \bar{\bm{z}}_k^f.
\end{equation}
The Kalman gain is then computed as
\begin{equation}
    K_k = P_{xz,k}^f \bigl(\widetilde{S}_k^f\bigr)^{-1}.
\end{equation}

When Improvements~1 and~2 are used together, EnKF recalibration uses the analogous recentered (recalibrated) mismatch
\begin{equation}
    \bm{d}_k^{rc}
    =
    h_k(\bar{\bm{x}}_k^a) - \bar{\bm{z}}_k^{rc},
\end{equation}
and the recalibrated innovation covariance becomes
\begin{equation}
    \widetilde{S}_k^{rc}
    =
    \frac{1}{N-1}\mathcal{Z}_k^{rc}(\mathcal{Z}_k^{rc})^\top
    +
    \beta_{k-1}\bm{d}_k^{rc}(\bm{d}_k^{rc})^\top
    +
    R_k.
\end{equation}
Accordingly, the recalibrated posterior covariance target becomes
\begin{equation}\label{hope3_comp}
    P_k^{a,rc}
    =
    P_k^f
    +
    K_k \widetilde{S}_k^{rc} K_k^\top
    -
    K_k (P_{xz,k}^{rc})^\top
    -
    P_{xz,k}^{rc} K_k^\top.
\end{equation}
In this case, the anomaly update should also be modified so that the stochastic EnKF preserves \eqref{hope3_comp} in expectation and the ETKF matches \eqref{hope3_comp} exactly. For the stochastic EnKF, the perturbed observations should be drawn from
\begin{equation}\label{rceta}
    \bm{\eta}_{k,i}^{rc}
    \sim
    \mathcal{N}\!\left(
        \bm{0},
        R_k + \beta_{k-1}\bm{d}_k^{rc}(\bm{d}_k^{rc})^\top
    \right),
\end{equation}
rather than from $\mathcal{N}(\bm{0},R_k)$, and the memberwise update becomes
\begin{equation}\label{eq:i2_stoch_temp}
    \bm{x}_{k,i}^{\star}
    =
    \bm{x}_{k,i}^{rc}
    +
    K_k\!\left(
        \bm{z}_k + \bm{\eta}_{k,i}^{rc} - \bm{z}_{k,i}^{rc}
    \right).
\end{equation}
By the same argument as in Proposition~\ref{prop:i1_stoch}, the compensated stochastic-EnKF update preserves the recalibrated target covariance in expectation after replacing $R_k$ by $R_k+\beta_{k-1}\bm{d}_k^{rc}(\bm{d}_k^{rc})^\top$. In particular, if the compensated update is accepted, then
\begin{equation}
    \mathbb{E}\!\left[\widetilde{P}_k^{a,rc}\right]
    =
    P_k^{a,rc},
\end{equation}
where $P_k^{a,rc}$ is defined in \eqref{hope3_comp}.

For the ETKF, the deterministic transform keeps the same form as \eqref{eq:mkrc}, but
\begin{equation}\label{Bnew}
    B_k^f
    =
    \Bigl(
        \mathcal{Z}_k^f(\mathcal{Z}_k^f)^\top
        +
        (N-1)\bigl(
            \beta_{k-1}\bm{d}_k^f(\bm{d}_k^f)^\top + R_k
        \bigr)
    \Bigr)^{-1},
\end{equation}
and
\begin{equation}\label{Gammanew}
    \Gamma_k^{rc}
    =
    \mathcal{Z}_k^{rc}(\mathcal{Z}_k^{rc})^\top
    +
    (N-1)\bigl(
        \beta_{k-1}\bm{d}_k^{rc}(\bm{d}_k^{rc})^\top + R_k
    \bigr)
\end{equation}
should replace their uncompensated counterparts. Likewise, Proposition~\ref{p2} extends directly to the compensated ETKF after replacing the uncompensated matrices $B_k^f$ and $\Gamma_k^{rc}$ by \eqref{Bnew} and \eqref{Gammanew}. Hence, if the compensated ETKF update is accepted, the realized sample covariance of the candidate analysis ensemble exactly matches the target covariance \eqref{hope3_comp}.

Mathematically, the covariance compensation proposed here is closely related to the second-order derivatives of the measurement function. Indeed, when each measurement component $h_{k,i}$ is quadratic with Hessian matrix $\mathcal{H}_i$, $i=1,\ldots,m$, the mismatch terms admit a closed-form expression:
\begin{equation}
    d_{k,i}^f
    =
    -\frac{N-1}{2N}\operatorname{tr}(\mathcal{H}_i P_k^f).
\end{equation}
Consequently,
\small\begin{equation}
    \bm{d}_k^f(\bm{d}_k^f)^\top
    =
    \frac{(N-1)^2}{4N^2}
    \bigl[\operatorname{tr}(\mathcal{H}_iP_k^f)\operatorname{tr}(\mathcal{H}_jP_k^f)\bigr]_{i,j}.
\end{equation}\normalsize
Therefore, the proposed covariance compensation scales automatically with the nonlinearity of the measurement function and becomes exactly zero when the measurements are linear.

The remaining task is to determine the scalar $\beta_{k-1}$, which controls the magnitude of the covariance compensation. Because the compensation is added directly to the innovation covariance, it is natural to tune $\beta_k$ online according to the discrepancy between the predicted and realized innovations. To this end, we use the normalized innovation squared (NIS),
\begin{equation}\label{NIS}
    \epsilon_k
    =
    \bigl(\bm{z}_k-\bar{\bm{z}}_k^f\bigr)^\top
    \bigl(\widetilde{S}_k^f\bigr)^{-1}
    \bigl(\bm{z}_k-\bar{\bm{z}}_k^f\bigr).
\end{equation}
To reduce sensitivity to instantaneous noise, we smooth the NIS by an exponentially weighted moving average:
\begin{equation}
    \bar{\epsilon}_k
    =
    \lambda \bar{\epsilon}_{k-1}
    +
    (1-\lambda)\epsilon_k,
\end{equation}
where $\lambda \in (0,1)$ is a user-chosen smoothing factor. We then update the compensation magnitude according to
\begin{equation}
    \beta_k
    =
    \max
    \Bigl(
        \beta_{k-1} + \mu(\bar{\epsilon}_k - m),\,
        0
    \Bigr),
    \label{eq:beta_update}
\end{equation}
where $\mu > 0$ is the adaptation gain, and $m$ is the measurement dimension. The projection onto $[0,\infty)$ ensures $\beta_k \ge 0$, so that the compensation term $\beta_k \bm{d}_k^f (\bm{d}_k^f)^\top$ remains positive semidefinite. The term $\bar{\epsilon}_k - m$ appears in \eqref{eq:beta_update} because, when the innovation statistics are consistent with the assumed covariance, the NIS satisfies
\begin{equation}
    \mathbb{E}[\epsilon_k] = m.
\end{equation}
Therefore, if the observed innovations are systematically larger than predicted, then $\bar{\epsilon}_k > m$ and $\beta_k$ increases, enlarging the innovation covariance. Conversely, if the observed innovations are systematically smaller than predicted, then $\bar{\epsilon}_k < m$ and $\beta_k$ decreases. Thus, the feedback law adaptively adjusts $\beta_k$ so that the predicted innovation covariance better matches the empirical innovation statistics.
 
\subsection{The proposed CAR-EnKF framework}
\label{subsec:universal_framework}
\begin{algorithm}[htbp]
\footnotesize
\caption{The proposed CAR-EnKF framework}
\label{alg:proposed_enkf}
\begin{algorithmic}[1]
    \Statex \textbf{Input:} process-noise covariance $Q_{k-1}$, measurement-noise covariance $R_k$, Ensemble size $N$, state-transition function $f_{k-1}$, measurement function $h_k$, input $\bm{u}_{k-1}$, measurement $\bm{z}_k$, compensation parameter $\beta_{k-1}$, EnKF variant $\in \{$stochastic EnKF, ETKF$\}$

    \Statex \textbf{Initialization:}
    \State Initialize the analysis ensemble $X_0^a$

    \For{$k=1,2,\dots$}

        \Statex \hspace{1.2em} \textbf{Predict:}
        \State Propagate $X_{k-1}^a$ through $f_{k-1}$ with $\bm{u}_{k-1}$ and $Q_{k-1}$ to obtain 
        \Statex \hspace{1.2em} the forecast ensemble $X_k^f$
        \State Compute the forecast mean $\bar{\bm{x}}_k^f$, anomaly matrix $A_k^f$, and covariance
        \[
            P_k^f = \frac{1}{N-1}A_k^f(A_k^f)^\top
        \]

        \Statex \hspace{1.2em} \textbf{Update:}
        \State Compute $\bm{z}_{k,i}^f = h_k(\bm{x}_{k,i}^f)$, $i=1,\dots,N$
        \State Form $\bar{\bm{z}}_k^f$, $\mathcal{Z}_k^f$, and $P_{xz,k}^f = \frac{1}{N-1}A_k^f(\mathcal{Z}_k^f)^\top$
        \State Compute the forecast nonlinearity mismatch
        \[
            \bm{d}_k^f = h_k(\bar{\bm{x}}_k^f) - \bar{\bm{z}}_k^f
        \]
        \State Form the compensated innovation covariance
        \[
            \widetilde{S}_k^f
            =
            \frac{1}{N-1}\mathcal{Z}_k^f(\mathcal{Z}_k^f)^\top
            +
            \beta_{k-1}\bm{d}_k^f(\bm{d}_k^f)^\top
            +
            R_k
        \]
        \State Compute the Kalman gain $K_k = P_{xz,k}^f(\widetilde{S}_k^f)^{-1}$
        \State Update the analysis mean $\bar{\bm{x}}_k^a
            =
            \bar{\bm{x}}_k^f
            +
            K_k(\bm{z}_k-\bar{\bm{z}}_k^f)$
        \Statex \hspace{1.2em} \textbf{Recalibrate:}
        \State Form the recentered ensemble $X_k^{rc} = \bar{\bm{x}}_k^a\bm{1}^\top + A_k^f$
        \State Compute $\bm{z}_{k,i}^{rc} = h_k(\bm{x}_{k,i}^{rc})$, $i=1,\dots,N$
        \State Form $\bar{\bm{z}}_k^{rc}$, $\mathcal{Z}_k^{rc}$, and $P_{xz,k}^{rc}
            =
            \frac{1}{N-1}A_k^f(\mathcal{Z}_k^{rc})^\top$
        \State Compute the recalibrated nonlinearity mismatch
        \[
            \bm{d}_k^{rc} = h_k(\bar{\bm{x}}_k^a) - \bar{\bm{z}}_k^{rc}
        \]
        \State Form the recalibrated compensated innovation covariance
        \[
            \widetilde{S}_k^{rc}
            =
            \frac{1}{N-1}\mathcal{Z}_k^{rc}(\mathcal{Z}_k^{rc})^\top
            +
            \beta_{k-1}\bm{d}_k^{rc}(\bm{d}_k^{rc})^\top
            +
            R_k
        \]
        \State Form the recalibrated posterior covariance target
        \[
            P_k^{a,rc}
            =
            P_k^f
            +
            K_k\widetilde{S}_k^{rc}K_k^\top
            -
            K_k(P_{xz,k}^{rc})^\top
            -
            P_{xz,k}^{rc}K_k^\top
        \]
        \State Form the candidate analysis ensemble $\widetilde{X}_k^a$
        \Statex \hspace{1.2em} -- For stochastic EnKF, use the perturbed-observation update
        \Statex \hspace{2em} defined by \eqref{rceta}, \eqref{eq:i2_stoch_temp}, and \eqref{Xk}--\eqref{translate}
        \Statex \hspace{1.2em} -- For ETKF, use the deterministic anomaly transform defined by
        \Statex \hspace{2em} \eqref{eq:mkrc}, \eqref{Bnew}, and \eqref{Gammanew}.
        \Statex \hspace{1.2em} \textbf{Back out:}
        \If{$\operatorname{tr}(P_k^{a,rc}) \le \operatorname{tr}(P_k^f)$}
            \State Accept the update and set $X_k^a = \widetilde{X}_k^a$
        \Else
            \State Reject the update and set $X_k^a = X_k^f$
        \EndIf

        \Statex \hspace{1.2em} \textbf{Optional adaptive compensation update:}
        \State Update $\beta_k$ using \eqref{NIS}--\eqref{eq:beta_update}, or set $\beta_k=\beta_{k-1}$ if a fixed 
        \Statex \hspace{1.2em} compensation is used
    \EndFor
\end{algorithmic}
\normalsize
\end{algorithm}
Algorithm~\ref{alg:proposed_enkf} summarizes CAR-EnKF. Most steps are shared across variants; only the anomaly update is variant-specific. Stochastic EnKF uses \eqref{rceta}, \eqref{eq:i2_stoch_temp}, and \eqref{Xk}--\eqref{translate}, whereas ETKF uses \eqref{eq:mkrc}, \eqref{Bnew}, and \eqref{Gammanew}.

\section{NUMERICAL RESULTS}\label{sec:exp}
We evaluate CAR-EnKF with the stochastic EnKF and ETKF on two nonlinear-observation benchmarks: feature-based simultaneous localization and mapping (SLAM) and Lorenz--96. For each variant and benchmark, we compare (i) the conventional EnKF, (ii) the EnKF with Improvement~1 only, and (iii) CAR-EnKF.

Unless otherwise stated, all simulations use \(1000\) Monte Carlo runs. The conventional EnKF and the EnKF with Improvement~1 use a fixed inflation factor \(\rho=1.05\), which is empirically effective and consistent with \cite{Gottwald2011VarianceLimiting}. To assess sensitivity to this choice, Figs.~\ref{slam2} and~\ref{Lorenz1} additionally show the conventional stochastic EnKF with \(\rho=1.0\) (no inflation) and \(\rho=1.1\). When Improvement~2 is enabled, inflation is disabled, and \(\beta_k\) is updated according to \eqref{NIS}--\eqref{eq:beta_update} with
\begin{equation}
    \beta_0 = 2,\quad
    \lambda = 0.9,\quad
    \mu = 0.1.
\end{equation}
Here, \(\beta_0\) sets the initial compensation magnitude, \(\lambda\) controls NIS memory (larger values smooth more), and \(\mu\) controls adaptation speed (smaller values reduce oscillation). We use the same values without retuning in all experiments.

For each benchmark, we report two experiments. In the fixed-parameter experiment, we plot the state RMSE (averaged over all Monte Carlo runs) versus time to illustrate the filters' transient and steady-state behavior. In the measurement-noise sweep, we plot the state RMSE (averaged over all Monte Carlo runs and all timesteps from step 11 onward) against a common measurement-noise scale factor to assess whether the proposed improvements remain effective across a broad range of noise levels.

\subsection{Feature-based SLAM Benchmark}
\label{subsec:exp_slam_setup}

\begin{figure}[htbp]
\centering
\includegraphics[height=4.2cm]{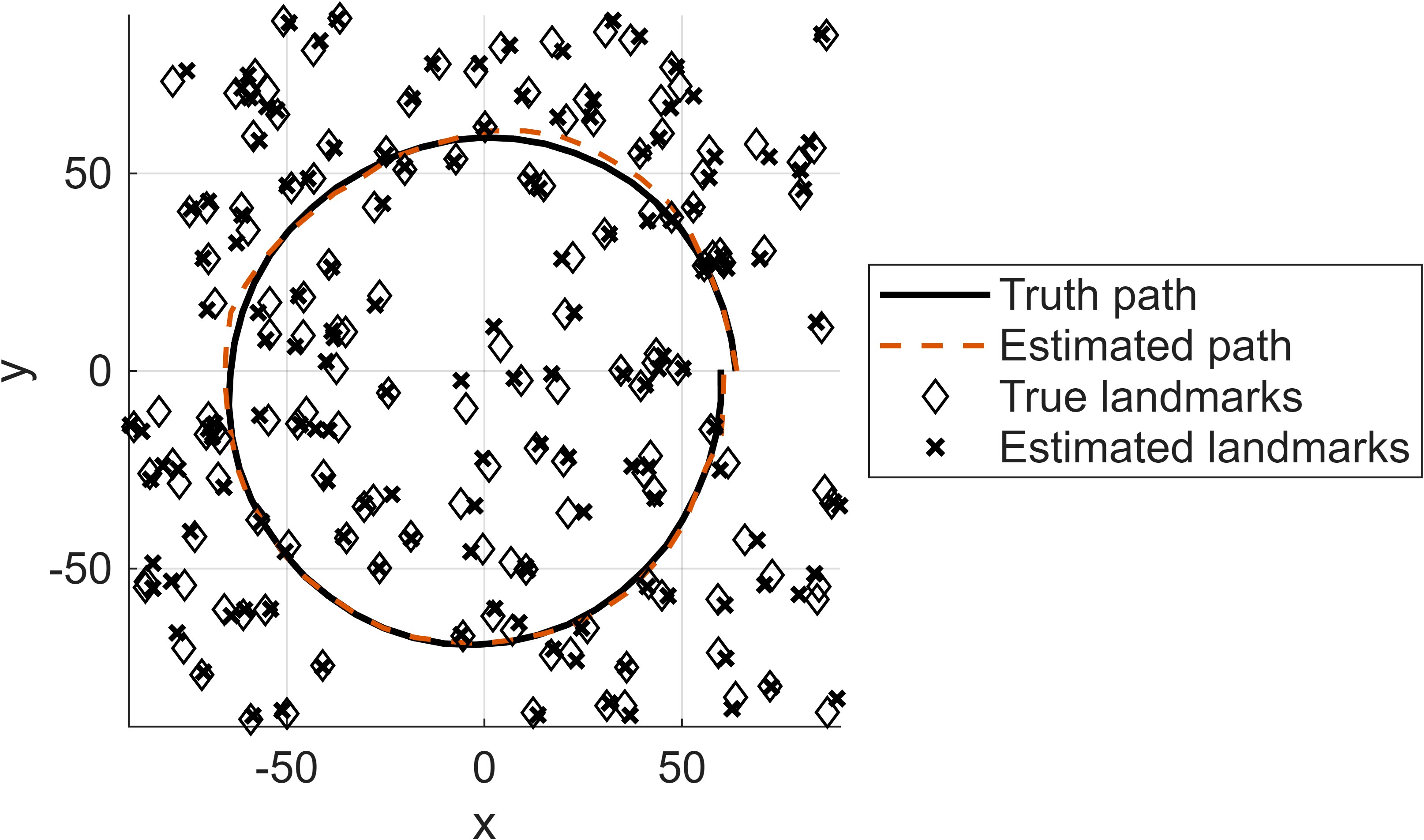}
\caption{Feature-based SLAM benchmark used in the paper.}
\label{slam1}
\end{figure}

We first evaluate the proposed framework on a \(2\)-D feature-based SLAM problem with nonlinear range--bearing measurements, following the benchmark structure of Sigges \emph{et al.}~\cite{sigges2018enkfslam}. To isolate the effect of the EnKF update rules, we assume known landmark correspondence and use a fixed-size SLAM state, thereby omitting the data-association and map-management components in \cite{sigges2018enkfslam}. The system setup is shown in Fig. \ref{slam1}. The state is
\begin{equation}
    \bm{x}_k =
    \begin{bmatrix}
        p_{x,k} & p_{y,k} & \theta_k & \ell_{1,x} & \ell_{1,y} & \cdots & \ell_{M,x} & \ell_{M,y}
    \end{bmatrix}^{\!\top},
\end{equation}
where \((p_{x,k},p_{y,k},\theta_k)\) is the robot pose and \(\ell_j=[\ell_{j,x},\,\ell_{j,y}]^\top\) is the \(j\)th landmark. The landmarks are static. Using the velocity-based motion model and range--bearing observations in \cite{sigges2018enkfslam}, we write
\begin{align}
    p_{x,k+1} &= p_{x,k} + \Delta t\, v_k \cos(\theta_k + \alpha_k) + w^x_k,\\
    p_{y,k+1} &= p_{y,k} + \Delta t\, v_k \sin(\theta_k + \alpha_k) + w^y_k,\\
    \theta_{k+1} &= \theta_k + \alpha_k + w^\theta_k,
\end{align}
where \(\bm{w}_k=[w_k^x,\,w_k^y,\,w_k^\theta]^\top\sim\mathcal{N}(\bm{0},Q)\). For each visible landmark \(j\),
\begin{equation}
    \bm{z}_k^{(j)} =
    \begin{bmatrix}
        \sqrt{(\ell_{j,x}-p_{x,k})^2 + (\ell_{j,y}-p_{y,k})^2} \\
        \mathrm{atan2}(\ell_{j,y}-p_{y,k},\,\ell_{j,x}-p_{x,k}) - \theta_k
    \end{bmatrix}
    + \bm{v}_k^{(j)},
\end{equation}
where \(\bm{v}_k^{(j)}\sim\mathcal{N}(\bm{0},R)\), and a landmark is observed only when its true range does not exceed the sensor range \(r_{\max}\).

We retain from \cite{sigges2018enkfslam} the circular-walk benchmark with \(M=150\) landmarks and \(r_{\max}=30\). The baseline noise covariances are
\small\begin{equation}
    Q=\mathrm{diag}(0.1^2,0.1^2,0.01^2), \quad
    R=\mathrm{diag}(0.1^2,0.01^2),
\end{equation}\normalsize
which correspond to \(\bm{\sigma}_w=[0.1,\,0.1,\,0.01]^\top\), \(\sigma_r=0.1\), and \(\sigma_\gamma=0.01\). The remaining parameters are \(\Delta t=1\), \(T=50\), \(N=100\), \(v_k=8\), \(\alpha_k=2\pi/T\), world half-width \(90\), robot-pose prior standard deviations \([2,\,2,\,15^\circ]\), and landmark-coordinate prior standard deviation \(8\). The landmarks are sampled uniformly in \([-90,90]^2\). In the measurement-noise sweep, we scale the baseline measurement-noise standard deviations by a common factor \(s\):
\begin{equation}
    \sigma_r = 0.1\,s,\quad
    \sigma_\gamma = 0.01\,s,\quad
    s \in 10^{-1:0.25:2}.
\end{equation}

\begin{figure}[htbp]
\centering
\includegraphics[height=3.9cm]{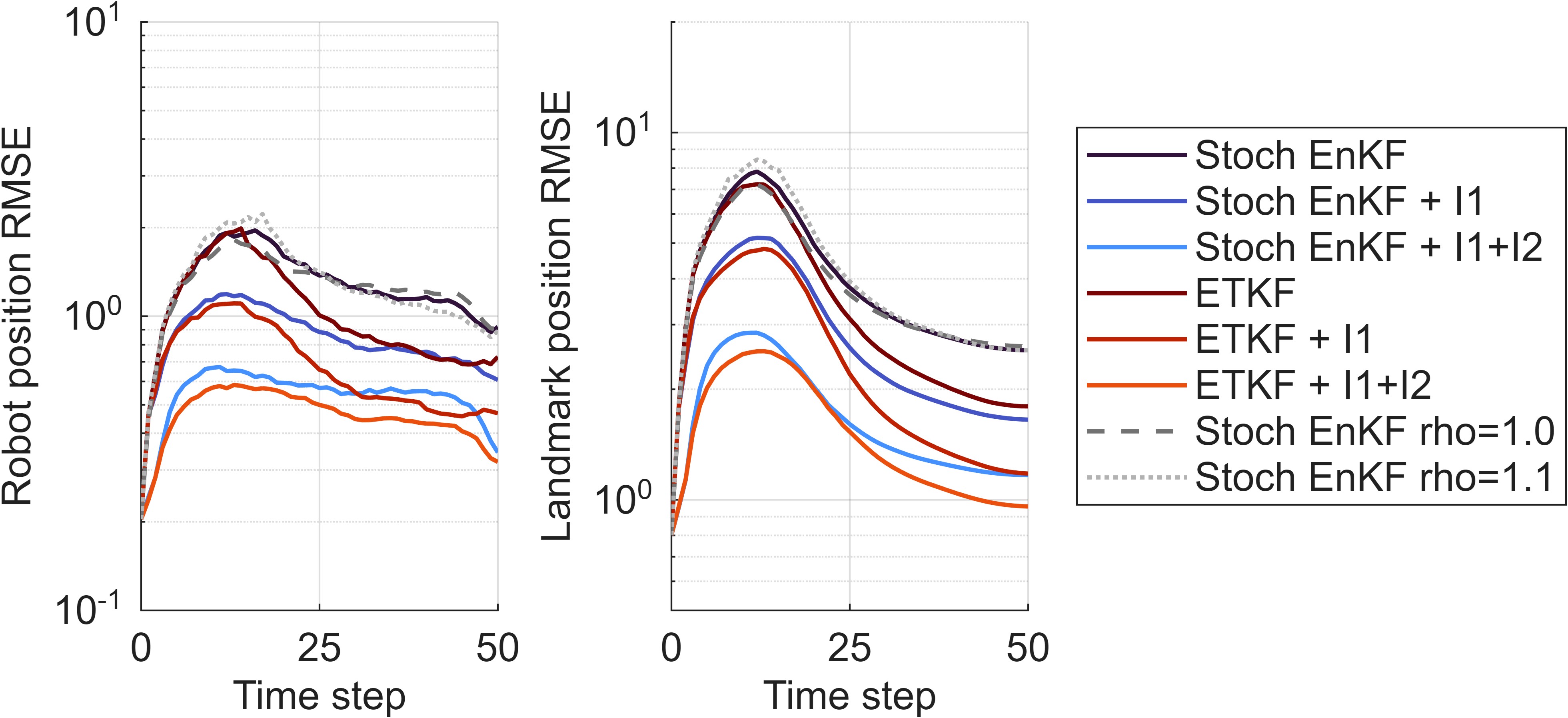}
\caption{SLAM RMSE with \(\sigma_r=0.1\) and \(\sigma_\gamma=0.01\). The conventional stochastic EnKF is shown for \(\rho=1.0,1.05,1.1\). I1 and I2 denote Improvements 1 and 2.}
\label{slam2}
\end{figure}
\begin{figure}[htbp]
\centering
\includegraphics[height=3.91cm]{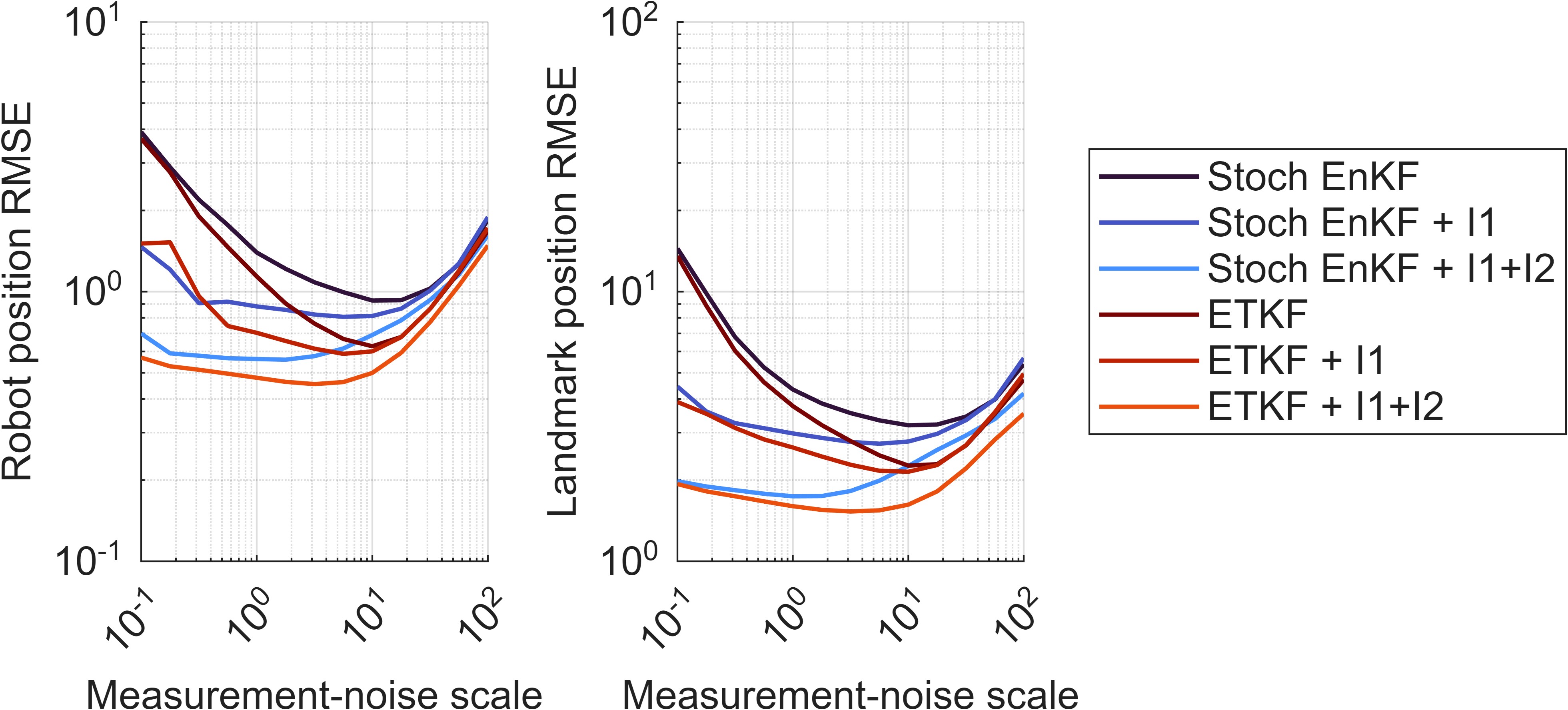}
\caption{Comparison of different EnKF algorithms on the SLAM problem under varying measurement-noise levels, with \(\sigma_r = 0.1\,s\) and \(\sigma_\gamma = 0.01\,s\), where $s$ denotes the scale factor. I1 and I2 denote the two improvements proposed in the previous section.}
\label{slam3}
\end{figure}
The fixed-parameter and measurement-noise-sweep results are shown in Figs.~\ref{slam2} and~\ref{slam3}. Relative to the default \(\rho=1.05\) baselines, Fig.~\ref{slam2} shows that Improvement~1 reduces the robot- and landmark-position RMSEs by about \(20\%\), while CAR-EnKF reduces them by about \(50\%\). The conventional stochastic-EnKF curves for \(\rho=1.0,1.05,1.1\) nearly overlap, showing that baseline performance is insensitive to \(\rho\) over the tested range and that tuning \(\rho\) does not explain the CAR-EnKF gain. Figure~\ref{slam3} further shows CAR-EnKF reductions of roughly \(10\%\) to \(90\%\), with larger gains when the measurements are more informative; the differences diminish as measurement noise increases.

It is also worth noting that the initial RMSE in Fig.~\ref{slam2} is much smaller than the prescribed prior standard deviation. This is because the prior ensemble is initialized around the true initial state, whereas the reported estimate is the ensemble mean. Therefore, if the ensemble members are sampled from a distribution with covariance \(P_0\) and the ensemble size is \(N\), then the ensemble-mean error has covariance approximately \(P_0/N\). As a result, the initial RMSE scales as \(1/\sqrt{N}\) relative to the prior standard deviation. The same phenomenon also appears in the Lorenz--96 benchmark discussed in the next subsection.

\subsection{Lorenz--96 Benchmark}
\label{subsec:exp_l96_setup}

We next evaluate the framework on Lorenz--96, a standard low-order surrogate widely used in atmospheric data assimilation and numerical weather prediction \cite{lorenz1998optimal,wu2014improving}. The state is
\begin{equation}
    \bm{x}(t)=
    \begin{bmatrix}
        x_1(t) & x_2(t) & \cdots & x_n(t)
    \end{bmatrix}^{\!\top},
    \quad n=40,
\end{equation}
with dynamics
\begin{equation}
    \frac{d x_i}{dt}
    =
    \left(x_{i+1}-x_{i-2}\right)x_{i-1} - x_i + F,
    \quad i=1,\ldots,n,
\end{equation}
where periodic indexing is used and \(F=8\). Following the standard setting in the literature, we use the 40-dimensional model with \(F=8\) and propagate the dynamics with a fourth-order Runge--Kutta integrator of step size \(\Delta t=0.05\) \cite{lorenz1998optimal,wu2014improving}.

We consider a perfect-model identical-twin experiment. The initial truth state is generated as \(F\bm{1}+\bm{\xi}\), where \(\bm{\xi}\sim\mathcal{N}(\bm{0},I)\), and is propagated for \(1000\) spin-up steps before data assimilation begins. The discrete-time forecast model is
\begin{equation}
    \bm{x}_{k+1} = \Phi_{\Delta t}(\bm{x}_k),
\end{equation}
where \(\Phi_{\Delta t}\) denotes one Runge--Kutta step. Following \cite{kurosawa2021data}, we use the nonlinear observation operator
\begin{equation}
    \bm{z}_k = h(\bm{x}_k) + \bm{v}_k,
    \quad
    h(\bm{x}_k) =
    \begin{bmatrix}
        x_{1,k}^2 & x_{3,k}^2 & \cdots & x_{39,k}^2
    \end{bmatrix}^{\!\top},
\end{equation}
where \(\bm{v}_k\sim\mathcal{N}(\bm{0},R)\) and \(R=\sigma_y^2 I\). Thus, \(m=20\) nonlinear measurements are available at each assimilation cycle.

The remaining parameters are \(N=50\), \(P_0=I\), \(120\) time steps (assimilation cycles) per Monte Carlo simulation, and baseline observation standard deviation \(\sigma_y=10^{-2}\). In the measurement-noise sweep, we vary only \(\sigma_y\) according to
\begin{equation}
    \sigma_y = 10^{-2}s,\quad
    s \in 10^{-1:0.25:2}.
\end{equation}

\begin{figure}[htbp]
\centering
\includegraphics[height=3.91cm]{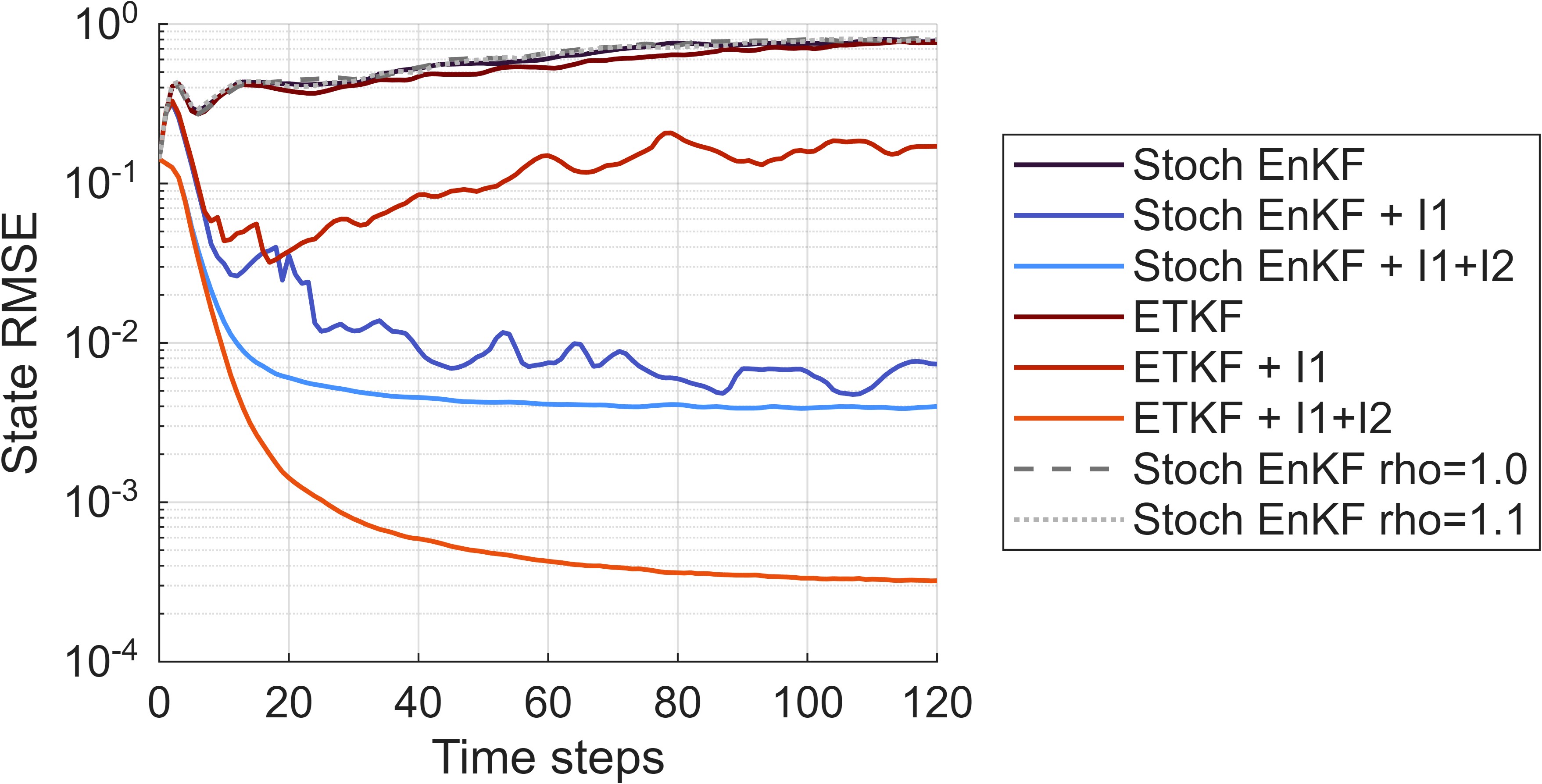}
\caption{Lorenz--96 RMSE with \(\sigma_y=10^{-2}\). The conventional stochastic EnKF is shown for \(\rho=1.0,1.05,1.1\). I1 and I2 denote Improvements 1 and 2.}
\label{Lorenz1}
\end{figure}
\begin{figure}[htbp]
\centering
\includegraphics[height=3.91cm]{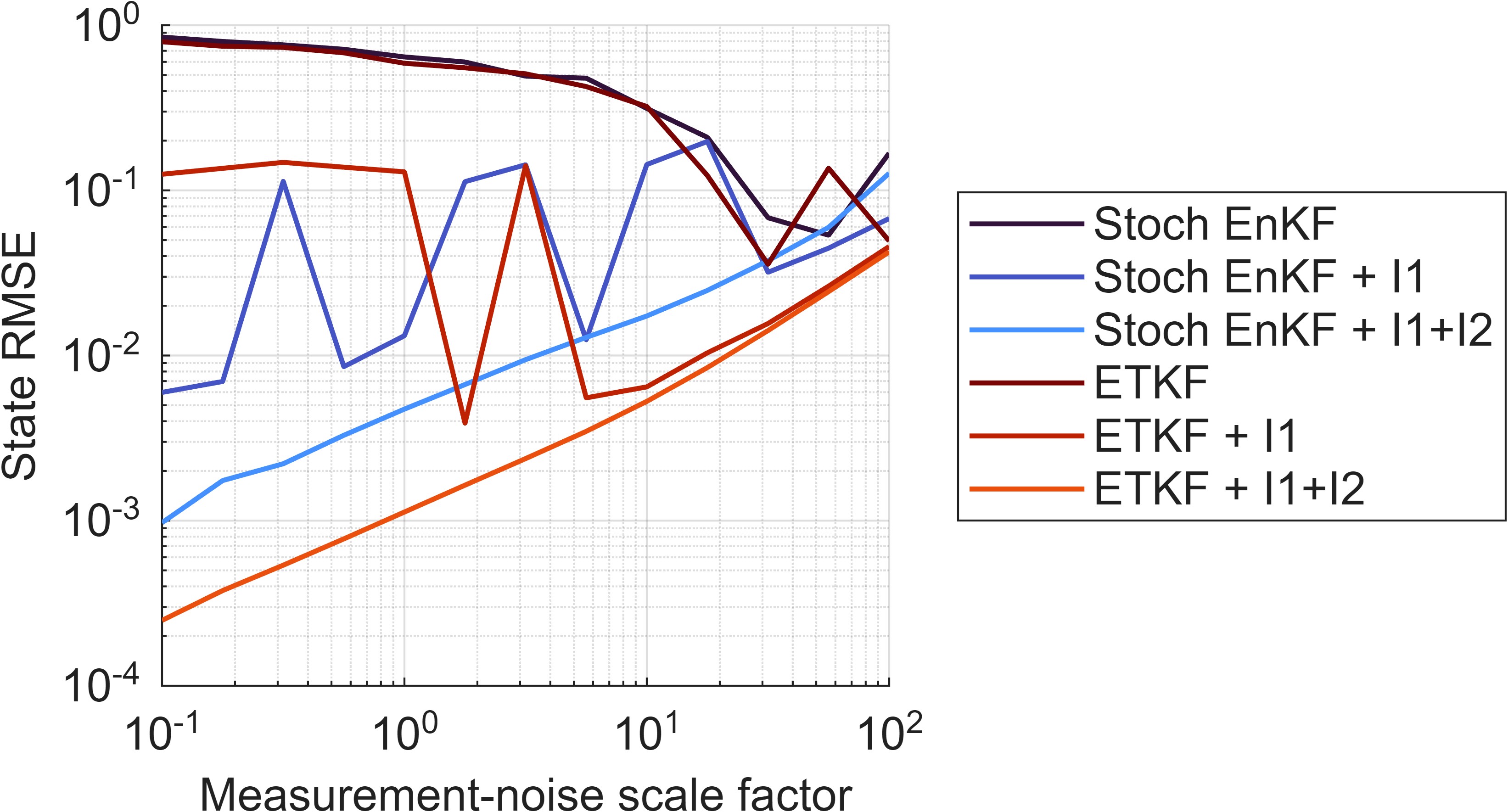}
\caption{Comparison of different EnKF algorithms on the Lorenz--96 system  under varying measurement-noise levels, with \(\sigma_y = 10^{-2}s\), where $s$ denotes the scale factor. I1 and I2 denote the two improvements proposed in the previous section.}
\label{Lorenz2}
\end{figure}

The fixed-parameter and measurement-noise-sweep results are shown in Figs.~\ref{Lorenz1} and~\ref{Lorenz2}. At \(\sigma_y=10^{-2}\), relative to the default \(\rho=1.05\) baselines, Improvement~1 reduces RMSE by about \(99\%\) for the stochastic EnKF and \(90\%\) for the ETKF, while CAR-EnKF reduces it by about \(99.5\%\) and \(99.9\%\), respectively. The conventional stochastic-EnKF curves for \(\rho=1.0,1.05,1.1\) nearly overlap, showing that changing the fixed inflation factor over this range does not resolve the estimation error. Across the sweep, CAR-EnKF reduces RMSE by roughly \(10\%\) to \(99.95\%\) relative to these tested baselines, with larger gains when the observations are more informative.

\section{CONCLUSIONS AND FUTURE WORK}\label{sec:conclusion}
This paper proposed CAR-EnKF, a covariance-adaptive and recalibrated ensemble Kalman filter framework for nonlinear state estimation. The proposed framework augments the conventional EnKF with a recalibration mechanism and a covariance-compensation term that responds directly to measurement nonlinearity. The algorithmic template is shared across EnKF variants, with only the anomaly update being variant-specific. We developed detailed realizations for the stochastic EnKF and the ETKF. Numerical results on feature-based SLAM and the Lorenz--96 system show that the proposed framework can substantially reduce state-estimation RMSE, especially when measurement error is low. Future work will compare against adaptive and iterative EnKFs and investigate consistency and stability guarantees.

\bibliographystyle{IEEEtran}
\bibliography{ref}

\end{document}